\documentclass[runningheads]{llncs}

\usepackage{amsmath}
\usepackage{amssymb}
\usepackage{xcolor}
\usepackage{graphicx}
\usepackage{hyperref}

\title{On Dasgupta's hierarchical clustering objective and its relation to other graph parameters} 

\titlerunning{On Dasgupta's hierarchical clustering objective} 

\author{Svein Høgemo \inst{1} \and Benjamin Bergougnoux \inst{1} \and Ulrik Brandes \inst{3} \and Christophe Paul \inst{2} \and Jan Arne Telle \inst{1}}

\institute{Department of Informatics, University of Bergen, Norway \and  LIRMM, CNRS, Univ Montpellier, France \and Social Networks Lab, ETH Zürich, Switzerland}

\authorrunning{S. Høgemo,\, B. Bergougnoux,\, U. Brandes,\, C. Paul\, and\, J. A. Telle}

\newcommand{\DCcost}{\mathsf{DC\text{-}value}}
\newcommand{\VPTsum}{\mathsf{VPT\text{-}sum}}
\newcommand{\VPTmax}{\mathsf{VPT\text{-}max}}
\newcommand{\EPTsum}{\mathsf{EPT\text{-}sum}}
\newcommand{\EPTmax}{\mathsf{EPT\text{-}max}}

\usepackage{thmtools}

\declaretheorem[name=Theorem]{theoremperso}
\declaretheorem[sibling=theoremperso, name=Lemma]{lemmaperso}
\declaretheorem[sibling=theoremperso,name=Claim]{claimperso}
\declaretheorem[sibling=theoremperso,name=Corollary]{corollaryperso}

\begin{document}

\maketitle

\begin{abstract}
The minimum height of vertex and edge partition trees are well-studied graph parameters
known as, for instance, vertex and edge ranking number.
While they are NP-hard to determine in general,
linear-time algorithms exist for trees.
Motivated by a correspondence with Dasgupta's objective for hierarchical clustering
we consider the total rather than maximum depth of vertices
as an alternative objective for minimization.
For vertex partition trees this leads to a new parameter 
with a natural interpretation as a measure of robustness against vertex removal.

As tools for the study of this family of parameters 
we show that they have similar recursive expressions
and prove a binary tree rotation lemma. 
The new parameter is related to trivially perfect graph completion 
and therefore intractable like the other three are known to be.
We give polynomial-time algorithms for both total-depth variants
on caterpillars and on trees with a bounded number of leaf neighbors.
For general trees, we obtain a 2-approximation algorithm.
\end{abstract}

\section{Introduction}

Clustering is a central problem in data mining and statistics. Although many objective functions have been proposed for (flat) partitions into clusters, hierarchical clustering has long been considered from the perspective of iterated merge (in agglomerative clustering) or split (in divisive clustering) operations.
In 2016, Dasgupta~\cite{Das16} proposed an elegant objective function, hereafter referred to as $\DCcost$,
for nested partitions as a whole, and thus facilitated the study of hierarchical clustering from an optimization perspective. This work has sparked research on other objectives, algorithms, and computational complexity, and drawn significant interest from the data science community~\cite{CKM+19}. 

 It is customary to represent the input data as an edge-weighted graph, where the weights represent closeness (in similarity clustering) or distance (in dissimilarity clustering).
 The bulk of work that has been done on $\DCcost$ has concentrated on assessing the performance of well-known clustering algorithms in terms of this objective. In Dasgupta's original paper, a simple divisive clustering algorithm for similarity clustering, recursively splitting the input graph along an $\alpha$-approximated sparsest cut, was shown to give a $O(\alpha\cdot\log n)$-approximation to $\DCcost$. In later papers, this result was further improved upon: Charikar and Chatziafratis~\cite{CC17} show that this algorithm in fact achieves an $O(\alpha)$-approximation of $\DCcost$, and complement this result with a hardness result for approximating $\DCcost$. They also provide new approximation algorithms by way of linear and semi-definite relaxations of the problem statement. The former is also pointed out by Roy and Pokutta~\cite{RoyPokutta}. For dissimilarity clustering (maximizing the objective function), several algorithms achieve constant approximation ratio, including average-linkage (the most commonly used agglomerative clustering method)~\cite{CKM+19}, although a semi-definite relaxation again can do a little better~\cite{charikar2019}.

In a recent paper showing that Dasgupta's objective remains intractable
even if the input dissimilarities are binary,
i.e., when hierarchically clustering an unweighted undirected graph,
Høgemo, Paul and Telle~\cite{hgemo_et_al:LIPIcs:2020:12713}
initiated the study of Dasgupta's objective as a graph parameter.
By the nature of Dasgupta's objective,
the associated cluster trees are binary, 
and admit a mapping from the inner nodes to the edges of the graph
such that every edge connects two vertices from different subtrees.
We relate such trees to so-called edge partition trees~\cite{IYER199143},
and show that minimizing Dasgupta's objective is equal to minimizing the total depth of all leaves in an edge partition tree.  

If we consider the maximum depth of a leaf (the height of the tree) instead, 
its minimum over all edge partition trees of a graph 
is known as the edge ranking number of that graph~\cite{IYER199143}.
The same concept applies to vertex partition trees,
in which there is a one-to-one correspondence between
all of its nodes (leaves and inner nodes) and the vertices of the graph
such that no edge connects two vertices whose corresponding nodes are in disjoint subtrees.
The minimum height of any vertex partition tree is called the
vertex ranking number~\cite{SCHAFFER198991,DEOGUN199939},
but also known as tree-depth~\cite{NESETRIL20061022,Nesetril015}
and minimum elimination tree height~\cite{Bodlaenderetal:rankings}.

\begin{table}[t]
\centering
\caption{A family of graph parameters based on nested graph decompositions.}
\label{tab:family}
\begin{tabular}{|l||l|c|}
\hline
\multicolumn{1}{|r||}{\qquad\qquad \textbf{vertex depth}}
       & \multicolumn{1}{c|}{\textbf{maximum}} & \textbf{total} \\
\textbf{partition tree}\qquad
       & \multicolumn{1}{c|}{\textbf{(max)}} & \textbf{(sum)} \\\hline\hline
       & & \\
\textbf{edge (EPT)}   & edge ranking number~\cite{IYER199143}
        & Dasgupta's clustering \\
       & & objective~\cite{Das16} \\\hline
       & tree-depth~\cite{NESETRIL20061022,Nesetril015}
        & \\
\textbf{vertex (VPT)} & vertex ranking number~\cite{SCHAFFER198991,DEOGUN199939}
        & [\emph{new in this paper}] \\
       & minimum elimination
        &    \\
        & tree height~\cite{Bodlaenderetal:rankings} & \\ \hline                      
\end{tabular}
\end{table}

The above places Dasgupta's objective, applied to unweighted graphs,
into a family of graph parameters as shown in Table~\ref{tab:family}.
It also suggests a new graph parameter,
combining the use of vertex partition trees
with the objective of minimizing the total depth of vertices.
All three previously studied parameters are NP-hard
to determine~\cite{pothen1988complexity,LAM199871,hgemo_et_al:LIPIcs:2020:12713}, 
and we show that the same holds for the new parameter. 
Interestingly, the proof relies on a direct correspondence
with trivially perfect graph completion
and thus provides one possible interpretation of the parameter
in terms of intersecting communities in social networks~\cite{Nastos}. 
We give an alternative interpretation
in terms of robustness against network dismantling.

For both parameters based on tree height,
efficient algorithms have been devised in case the input graph is a tree.
For the edge ranking number, it took a decade 
from a polynomial-time $2$-approximation~\cite{IYER199143} and
an exact polynomial-time algorithm~\cite{de1995optimal}
to finally arrive at a linear-time algorithm~\cite{LamYue2001}.
Similarly, a polynomial-time algorithm for the vertex ranking number~\cite{IYER1988225}
was later improved to linear time~\cite{SCHAFFER198991}.
No such algorithms for the input graph being a tree are known for the total-depth variants.

Our paper is organized as follows. In Section 2 we give formal definitions, and give a rotation lemma for general graphs to improve a given clustering tree. This allows us to show that if a clustering tree for a connected graph has an edge cut which is not minimal, or has a subtree defining a cluster that does not induce a connected subgraph, then it cannot be optimal for $\DCcost$.
In Section 3 we go through the 4 problems in
Table~\ref{tab:family} and prove the equivalence with the standard definitions. We also show an elegant and useful recursive formulation of each of the 4 problems. 
In Section 5 we consider the situation when the input graph is a tree.
We give polynomial-time algorithms to compute the total depth variants, including $\DCcost$, for caterpillars and more generally for trees having a bounded number of leaves in the subtree resulting from removing its leaves.
We then consider the sparsest cut heuristic used by Dasgupta~\cite{Das16}
to obtain an approximation on general graphs.
When applied to trees, even to caterpillars, this does not give an optimal algorithm for $\DCcost$. However, we show that it does give a 2-approximation on trees, which improves on an 8-approximation due to Charikar and Chatziafratis \cite{charikar2019}. 

We leave as open the question if any of the two total depth variants can be solved in polynomial time on trees. On the one hand it would be very surprising if a graph parameter with such a simple formulation was NP-hard on trees. On the other hand, the graph-theoretic footing of the algorithms for the two max depth variants on trees does not seem to hold. The maximum depth variants are amenable to greedy approaches, where any vertex or edge that is reasonably balanced can be made root of the partition tree, while this is not true for the total depth variants.

\section{Preliminaries}\label{sec:defs}

We use standard graph theoretic notation \cite{Diestel_2017}.
In this paper, we will often talk about several different trees: an unrooted tree which is the input to a problem, and a rooted tree which is a decompositional structure used in the problem formulation. To differentiate the two, we will denote an unrooted tree as any graph, $G$, while a rooted tree is denoted $T$. Furthermore, $V(G)$ are called the \emph{vertices} of $G$, while $V(T)$ are called the \emph{nodes} of $T$.

A rooted tree has the following definition; a tree (connected acyclic graph) $T$ equipped with a special node called the \emph{root} $r$, which produces an ordering on $V(T)$. Every node $v$ in $T$ except $r$ has a \emph{parent}, which is the neighbor that lies on the path from $v$ to $r$. Every node in $T$ that has $v$ as its parent is called a \emph{child} of $v$. A node with no children is called a \emph{leaf}. Leaves are also defined on unrooted trees as vertices which have only one neighbor. The set of leaves in a tree $T$ is denoted $L(T)$. The subtree induced by the internal vertices of $T$, i.e. $T\setminus L(T)$, is called the \emph{spine-tree} of $T$. A \emph{caterpillar} is a tree Whose spine-tree is a path; this is the \emph{spine} of the caterpillar.

In a rooted tree, the set of nodes on the path from $v$ to $r$ is called the \emph{ancestors} of $v$, while the set of all nodes that include $v$ on their paths to $r$ is called the \emph{descendants} of $v$. We denote by $T[v]$ the subtree induced by the descendants of $v$ (naturally including $v$ itself).
As can be seen already for the paragraph above, we reserve the name \emph{node} for the vertices in rooted trees. In unrooted trees and graphs in general we only use \emph{vertex}; this is to avoid confusion.
For a given graph $G$, we use $n(G)$ and $m(G)$ to denote $|V(G)|$ and $|E(G)|$, respectively, or simply $n$ and $m$ if clear from context.
Let $A$ be a subset of $V(G)$. Then $G[A]$ is the \emph{induced subgraph} of $G$ by $A$, i.e. the graph $(A,\{uv\in G \mid u,v\in A\})$. If $B$ is a subset of $V(G)$ disjoint from $A$, then $G[A,B]$ is the bipartite subgraph of $G$ induced by $A$ and $B$, i.e. the graph $(A\cup B,\{uv\in G \mid u\in A \land v\in B\})$.
A \emph{cut} in a graph is a subset of the edges that, if removed, leaves the graph disconnected.
If $G$ is an unrooted tree, then every single edge $uv$ forms a cut, and we let $G_u$ (respectively $G_v$) denote the connected component of $G-uv$ containing $u$ (respectively $v$).
We use $[k]$ to denote the set of integers from 1 to $k$.

\begin{definition}[Edge-partition tree, Vertex-partition tree]
Let $G$ be a connected graph. An \emph{edge-partition tree} $T$ of $G$ is a rooted tree where:
\begin{itemize}
    \item The leaves of $T$ \emph{are} $V(G)$ and the internal nodes of $T$ \emph{are} $E(G)$. 
    \item Let $r$ be the root of $T$. If $G' = G-r$ has $k$ connected components $G'_1,\ldots,G'_k$ (note that $k\leq 2$), then $r$ has $k$ children $c_1,\ldots,c_k$.
    \item For all $1\leq i \leq k$, $T[c_i]$ is an edge partition tree of $G'_i$.
\end{itemize}
A \emph{vertex-partition tree} $T$ of $G$ is a rooted tree where:
\begin{itemize}
    \item The  nodes of $T$ \emph{are} $V(G)$. 
    \item Let $r$ be the root of $T$. If $G' = G-r$ has $k$ connected components $G'_1,\ldots,G'_k$, then $r$ has $k$ children $c_1,\ldots,c_k$.
    \item For all $1\leq i \leq k$, $T[c_i]$ is a vertex partition tree of $G'_i$.
\end{itemize}
The set of all edge partition trees of $G$ is denoted $EPT(G)$ and the set of all vertex partition trees of $G$ is denoted $VPT(G)$.
\end{definition}

For each node $x$ in an edge partition tree $T$, we denote by $ed_T(x)$ the \emph{edge depth} of $x$ in $T$, i.e. the number of tree edges on the path from the root of $T$ to $x$.
For each node $x$ in a vertex partition tree $T$, we denote by $vd_T(x)$ the \emph{vertex depth} of $x$ in $T$, equal to $ed_T(x)+1$, i.e. the number of nodes on the path from the root of $T$ to $x$.
The \emph{vertex height} of a tree $T$ is equal to the maximum vertex depth of the nodes in $T$, and the \emph{edge height} of a tree $T$ is equal to the maximum edge depth of the nodes in $T$.
We generally assume that the graph $G$ is connected; if $G$ is disconnected, then any VPT (or EPT) is a forest, consisting of the union of VPTs (EPTs) of the components of $G$.

\begin{figure}[tb]
    \centering
    \includegraphics[width=0.9\textwidth]{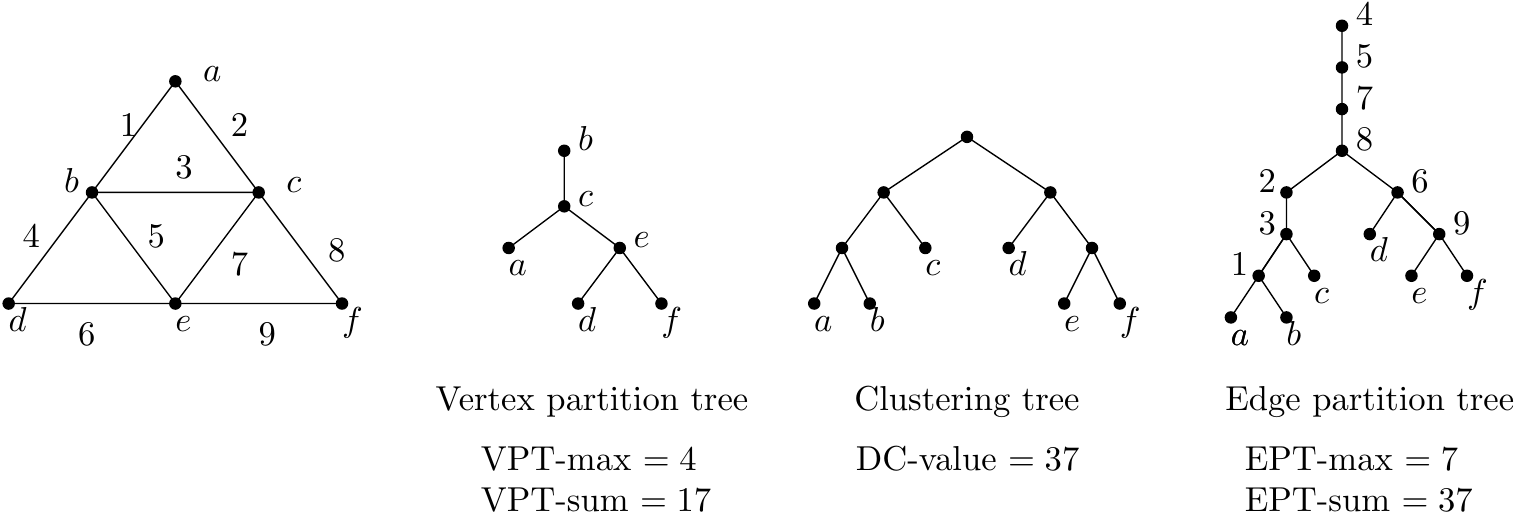}
    \caption{This figure shows the different types of partition trees on a small graph (the 3-sun). Vertices are marked with letters and edges with numbers. The clustering tree and the edge partition tree have the same structure. All trees are optimal for the measures defined in Section \ref{Sec:4prob}.}
    \label{fig:partitiontrees}
\end{figure}

A graph $G$ is \emph{trivially perfect} if there is a vertex partition tree $T$ of $G$ such that for any two vertices $u,v$, if $u$ is an ancestor of $v$ (or vice versa) in $T$, then $uv$ is an edge in $G$. We call $T$ a \emph{generating tree} for $G$.
Trivially perfect graphs are also known as \emph{quasi-threshold graphs} or \emph{comparability graphs of trees} (see \cite{JINGHO1996247}).


\begin{definition}[Clustering tree] A binary tree $T$ is a clustering tree of $G$ if:
\begin{itemize}
    \item The leaves of $T$ are $V(G)$. The clustering tree of $K_1$ is that one vertex.
    \item Let $r$ be the root of $T$, with children $a$ and $b$. Then $A=L(T[a])$ and $B=L(T[b])$ is a partition of $V(G)$.
    \item $T[a]$ and $T[b]$ are clustering trees of $G[A]$ and $G[B]$, respectively.
\end{itemize}
\end{definition}

 For any node $x\in T$, we define $G[x]$ as shorthand for $G[L(T[x])]$, and for two siblings $a,b\in T$ we define $G[a,b]$ as shorthand for $G[L(T[a]),L(T[b])]$.

\begin{definition}[$\DCcost$]
The Dasgupta Clustering value of a graph $G$ and a clustering tree $T$ of $G$ is
$$\DCcost(G,T) =\sum_{x\in V(T)\setminus L(T)} m(G[a_x,b_x])\cdot n(G[x])$$
where $a_x$ and $b_x$ are the children of $x$ in $T$. The $\DCcost$ of $G$, $\DCcost(G)$, is the minimum $\DCcost$ over all of its clustering trees.
\end{definition}

The following lemma  gives a condition under which one can improve a given hierarchical clustering tree by performing either a left rotation or a right notation at some node of the tree. See Figure \ref{fig:rotate}. First off, it is easy to see that performing such a rotation maintains the property of being a clustering tree. 

\begin{lemmaperso}[Rotation Lemma] \label{lemma:rotate}
Let $G$ be a graph with clustering trees $T$ and $T'$ such that $T'$ is the result of left rotation in $T$ and $T$ of a right rotation in $T'$. Let $T$ and $T'$ have nodes $a,b,c,t,u$ as in Figure \ref{fig:rotate}. We have
$$\DCcost(T)-\DCcost(T')= n(G[c]) \cdot m(G[a,b]) - n(G[a]) \cdot m(G[b,c]).$$
\end{lemmaperso} 

\begin{proof}
The $\DCcost$ of $T[u]$ is equal to
\begin{align*}(n(G[a]) + n(G[b]) &+ n(G[c])) \cdot (m(G[a,b]) + m(G[a,c])) + (n(G[b]) + n(G[c])) \cdot\\ m(G[b,c]) +
&\DCcost(T[a]) + \DCcost(T[b]) + \DCcost(T[c])
\end{align*}
and the DC-cost of the rotated tree $T'[t]$ is equal to
\begin{align*}
(n(G[a]) + n(G[b]) &+ n(G[c])) \cdot (m(G[a,c]) + m(G[b,c])) + (n(G[a]) + n(G[b])) \cdot m(G[a, b]) +\\
&\DCcost(T[a]) + \DCcost(T[b]) + \DCcost(T[c])
\end{align*}
See Figure \ref{fig:rotate} for reference. By substituting the costs written above and
cancelling out, we get the equality in the statement of the lemma.
 \qed\end{proof}
 
\begin{figure}[tb]
	\centering
	\includegraphics[width=0.7\linewidth]{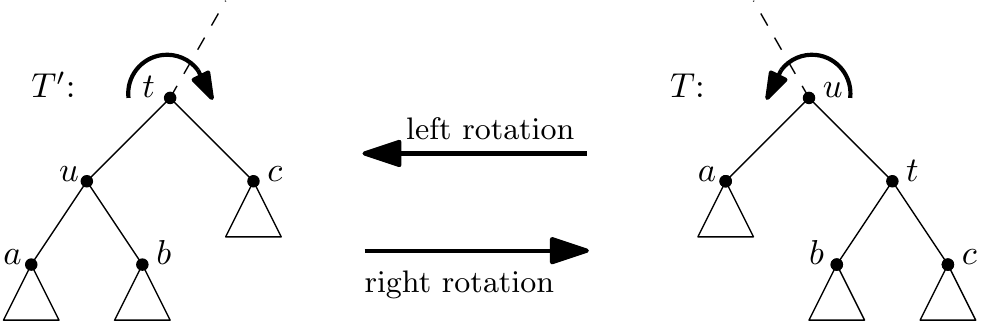}
	\caption{Here $T'$ is derived from $T$ by a left tree rotation, or equivalently $T$ is derived from $T'$ by a right tree rotation.}
	\label{fig:rotate}
\end{figure}

This lemma proves useful anywhere where we would like to manipulate clustering trees. We first use it to prove an important fact about $\DCcost$:

\begin{theoremperso} \label{thm:conn}
Let $G$ be a connected graph, and let $T$ be an optimal hierarchical clustering of $G$. Then, for any node $t \in T$, the subgraph $G[t]$ is connected.
\end{theoremperso}
\begin{proof}
We assume towards a contradiction that there exists a connected graph $G$ and an optimal hierarchical clustering $T$ of $G$, with some node $t \in T$ such that the subgraph $G[t]$ is not connected. Observe that for $r$, the root in $T$, $G[r] = G$ is connected. Then there must exist a node $t'$ such that $G[t']$ is not connected and for every ancestor $u \neq t'$ (of which there is at least one) $G[u]$ is connected. We focus on $t'$, its parent $u'$, its children $b$ and $c$, and its sibling, $a$. The following claim is useful:
\begin{claimperso}[\cite{Das16}, Lemma 2]\label{claim:disconn}
Let $G$ be a disconnected graph. In an optimal clustering tree of $G$, the cut induced by the root is an empty cut.
\end{claimperso}
Since $T$ is optimal, by Claim \ref{claim:disconn} there are no edges going between the subgraphs $G[b]$ and $G[c]$ in $G$. Since $G[u']$ is connected, there must be at least one edge going between $G[a]$ and $G[b]$ in $G$. We thus have
$n(G[a]) \cdot m(G[b,c]) = 0$
and
$n(G[c]) \cdot m(G[a,b]) > 0$.
But now, by Lemma \ref{lemma:rotate}, we can perform a tree rotation on $T$ to obtain a clustering with strictly lower cost than $T$. This implies that $T$ cannot be optimal after all. Thus, the theorem is true as stated.
\qed\end{proof}

\begin{corollaryperso}\label{cor:cut}
Let $T$ be an optimal clustering tree of a graph $G$ (not necessarily connected). Then, for every internal node $t \in T$ with children $u,v$, the cut $E(G[u,v])$ is an inclusion-wise minimal cut in $G[t]$.
\end{corollaryperso}

That all optimal clustering trees have this property is hardly surprising, but is still worth pointing out. It is hard to imagine a scenario where this property would be unwanted in an application of similarity-based hierarchical clustering. Also, going forward in this paper, we will be exclusively working with this kind of clustering trees. Therefore we give it the name:

\begin{definition}[Viable clustering tree]
Let $T$ be a clustering tree of some graph. We say that $T$ is a viable clustering tree if it has the added restriction that for every internal node $x\in T$ with children $a_x,b_x$, the cut induced by the partition $(L(T[a_x]),L(T[b_x]))$ is an inclusion-wise minimal cut in $G[L(T[x])]$.
\end{definition}


\section{Four Related Problems}\label{Sec:4prob}

We define four measures on partition trees of a graph $G$, three of them well-known in the literature. To give a unified presentation, throughout this paper we will call them $\VPTsum$, $\VPTmax$, $\EPTsum$ and $\EPTmax$, with no intention to replace the more well-known names.
All four measures can be defined with very simple recursive formulas.


\begin{definition}
$\VPTmax(G)$ is the minimum vertex height over trees $T\in VPT(G)$.
\end{definition}

This is arguably the most well-known of the four measures. It is known under several names, such as \emph{tree-depth}, \emph{vertex ranking number}, and \emph{minimum elimination tree height}.
The definition of tree-depth and minimum elimination tree height is exactly the minimum height of a vertex partition tree (an elimination tree is a vertex partition tree).
The equivalence of vertex ranking number and minimum elimination tree height is shown in \cite{10.1007/3-540-57785-8_187}, while it is known from \cite{NESETRIL20061022} that 
$$\VPTmax(G) = \min_{v\in V(G)}\ (1 + \max_{C\in cc(G-v)}\ \VPTmax(C)).$$


\begin{definition}
$\EPTmax(G)$ is the minimum edge height over trees $T\in EPT(G)$.
\end{definition}

It is known that $\EPTmax(G)$ is equivalent to the edge ranking number~\cite{IYER199143}.

\begin{theoremperso}\label{thm:EPTmax}
For any connected graph $G$, \[\EPTmax(G) = \min_{e\in E(G)}\ (1 + \max_{C\in cc(G-e)}\ \EPTmax(C)).\]
\end{theoremperso}
\begin{proof}
For $T\in EPT(G)$, we denote $\EPTmax(G,T)=\max_{\ell\in L(T)} ed_T(\ell)$. Let $T^*$ be an optimal EPT of $G$, that is 
$$\EPTmax(G)=\min_{T\in EPT(G)} \EPTmax(G,T)=\EPTmax(G,T^*).$$
Since $G$ is connected, $T$ has only one root $r$. We let $c_1,\dots,c_k$ be the children of $r$.
We denote by $T^*_i$ the subtree $T^*[c_i]$ and by  $C_i$ the induced subgraph $G[L(T^*_i)]$. Observe that we have:
  \[\begin{array}{rl}
   \EPTmax(G,T^*)=& \max_{\ell\in L(T^*)} ed_{T^*}(\ell)\\
  = & 1 + \max_{i\in[1,k]} \max_{\ell_i\in L(T^*i)} ed_{T^*_i}(\ell_i)\\
  = & 1 + \max_{i\in[1,k]} \EPTmax(C_i,T^*_i).
  \end{array}\]
Suppose that $e_r$ is the edge of $G$ mapped to $r$ in $T^*$. 
As $T^*$ is optimal for at least one $i\in [1,k]$ we have $\EPTmax(G,T^*)= 1 + \EPTmax(C_i,T^*_i)=1+\EPTmax(C_i)$ and for every $j\neq i$ we have $\EPTmax(C_j)\leq \EPTmax(C_i)$. By definition of EPT's, it follows that:
\[\begin{array}{rl}
   \EPTmax(G,T^*)=& 1 + \max_{i\in[1,k]} \EPTmax(C_i)\\
   =& 1 + \max_{C\in cc(G-e_r)} \EPTmax(C).
  \end{array}\]
  Again the optimality of $T^*$ implies that
  $$ 1 + \max_{C\in cc(G-e_r)} \EPTmax(C)=\min_{e\in E(G)} (1 + \max_{C\in cc(G-e)} \EPTmax(C)),$$
concluding the proof of the theorem. \qed\end{proof}

\begin{definition}\label{def:EPTsum}
$\EPTsum(G)$ is the minimum over every tree $T\in EPT(G)$ of the sum of the edge depth of all \emph{leaves} in $T$.
\end{definition}

The following equivalence between $\EPTsum$ and $\DCcost$, and the very simple recursive formula, forms the motivation for the results we present here.

\begin{theoremperso}\label{thm:eptsum}
For any connected graph $G$, $\EPTsum(G)=\DCcost(G)$ and 
$$\EPTsum(G) = \min_{e\in E(G)}\ (n(G) + \sum_{C\in cc(G-e)}\ \EPTsum(C)).$$
\end{theoremperso}
\begin{proof}
We begin proving the equivalence between $\EPTsum$ and its recursive formulation.
For $T\in EPT(G)$, we denote $\EPTsum(G,T)=\sum_{\ell\in L(T)} ed_T(\ell)$. Let $T^*$ be an optimal EPT of $G$.
Since $G$ is connected, $T^*$ has only one root $r$. We let $c_1,\dots,c_k$ be the children of $r$.
We denote by $T^*_i$ the subtree $T^*[c_i]$ and by  $C_i$ the induced subgraph $G[L(T^*_i)]$. Observe that we have:
  \[\begin{array}{rl}
   \EPTsum(G,T^*)=& \sum_{\ell\in L(T^*)} ed_{T^*}(\ell)\\
  = & n(G) + \sum_{i\in[1,k]} \sum_{\ell_i\in L(T^*i)} ed_{T^*_i}(\ell_i)\\
  = & n(G) + \sum_{i\in[1,k]} \EPTsum(C_i,T^*_i).
  \end{array}\]
Suppose  $e_r$ is the edge of $G$ mapped to $r$ in $T^*$. As $T^*$ is optimal, for every $i\in[1,k]$, $T^*_i$ is an optimal EPT of $C_i$. We have
\[\begin{array}{rl}
   \EPTsum(G,T^*)=& n(G) + \sum_{i\in[1,k]} \EPTsum(C_i)\\
   =& n(G) + \sum_{C\in cc(G-e_r)} \EPTsum(C)\\
   =& \min_{e\in E(G)} (n(G) + \sum_{C\in cc(G-e)} \EPTsum(C))
  \end{array}\]
  where the first two equalities follow from the definition of EPT's and the last from the optimality of $T^*$, and we  conclude that the recursive formula holds.

Now, we prove the equivalence between $\DCcost$ and $\EPTsum$.
Given a clustering tree $CT$ of a graph $G$, which by Theorem~\ref{cor:cut} can be assumed to be viable, it is easy to construct an EPT $T$ such that $\DCcost(G,CT)=\EPTsum(G,T)$. For each internal node $x$ of $CT$ with children $a_x,b_x$, we replace $x$ with a path $P_x$ on $m(G[a_x,b_x])$ nodes, connect one end to the parent of $x$ and the other end to the two children. Then we construct an arbitrary map between the nodes on the path $P_x$ and the edges in $G[a_x,b_x]$. As $CT$ is viable, $E(G[a_x,b_x])$ is an inclusion-wise minimal cut of $G[x]$ and thus $T$ is an EPT of $G$.
We have $\DCcost(G,CT)=\EPTsum(G,T)$ because when we replace $x$ by the path $P_x$, we increase the edge depth of the $n(G[x])$ leaves of the subtree rooted at $x$ by $m(G[a_x,b_x])$.
Conversely, given an EPT $T$ of $G$, contracting every path with degree-two internal nodes into a single edge results in a clustering tree $CT$ (not necessarily viable unless $T$ is optimal) and we have $\DCcost(G,CT)=\EPTsum(G,T)$ by the same argument as above.
We conclude from these constructions that $\DCcost(G)=\EPTsum(G)$. \qed\end{proof}


\begin{definition}
\label{def:vptsum}
$\VPTsum(G)$ is the minimum over every tree $T\in VPT(G)$ of the sum of the vertex depth of all nodes in $T$.
\end{definition}

\begin{theoremperso} \label{thm:VPTsum}
For any connected graph $G$, we have \[\VPTsum(G) = \min_{v\in V(G)}\ (n(G) + \sum_{C\in cc(G-v)}\ \VPTsum(C)).\]
\end{theoremperso}
\begin{proof}
For $T\in VPT(G)$, we denote $\VPTsum(G,T)=\sum_{v\in V(T)} vd_T(v)$. Let $T^*$ be an optimal VPT of a connected graph $G$, that is 
$$\VPTsum(G)=\min_{T\in VPT(G)} \VPTsum(G,T)=\VPTsum(G,T^*).$$
Since $G$ is connected, $T$ has only one root $r$. We let $c_1,\dots,c_k$ be the children of $r$.
We denote by $T^*_i$ the subtree $T^*[c_i]$ and by  $C_i$ the induced subgraph $G[L(T^*_i)]$. Observe that we have:
  \[\begin{array}{rl}
   \VPTsum(G,T^*)=& \sum_{v\in V(T^*)} vd_{T^*}(v)\\
  = & n(G) + \sum_{i\in[1,k]} \sum_{v_i\in V(T^*i)} vd_{T^*_i}(v_i)\\
  = & n(G) + \sum_{i\in[1,k]} \VPTsum(C_i,T^*_i).
  \end{array}\]
Suppose that $v_r$ is the vertex of $G$ mapped to $r$ in $T^*$. As $T^*$ is optimal, for every $i\in[1,k]$, $T^*_i$ is an optimal VPT of $C_i$. And by definition of VPT's, it follows that:
\[\begin{array}{rl}
   \VPTsum(G,T^*)=& n(G) + \sum_{i\in[1,k]} \VPTsum(C_i)\\
   =& n(G) + \sum_{C\in cc(G-v_r)} \VPTsum(C).
  \end{array}\]
  Again the optimality of $T^*$ implies that
  $$ n(G) + \sum_{C\in cc(G-v_r)} \VPTsum(C)=\min_{v\in V(G)} (n(G) + \sum_{C\in cc(G-v)} \VPTsum(C)),$$
concluding the proof of the theorem.
\qed \end{proof}

When comparing the definition of $\VPTsum$ with the definition of trivially perfect graphs,
it is not hard to see that a tree minimizing $\VPTsum(G)$ is a generating tree of a trivially perfect supergraph of $G$ where as few edges as possible have been added.

\begin{theoremperso}\label{thm:VPTeqTPC}
For any graph $G$, there exists a trivially perfect completion of $G$ with at most $k$ edges iff $\VPTsum(G) \leq k+n(G)+m(G)$.
\end{theoremperso}
\begin{proof}
Let $G'$ be a trivially perfect graph, and let $T$ be a generating tree for $G'$. As two vertices in $G'$ are adjacent if and only if the corresponding nodes in $T$ have an ancestor-descendant relationship, we can find the number of edges in $G'$ by summing up, for each node $v$ in $T$, the number of ancestors of $v$. This number is clearly equal to $vd_T(v)-1$. Therefore, $\VPTsum(G') = m(G')+n(G')$. Thus, a subgraph $G$ on the same vertex set has a completion into $G'$ on at most $k$ edges if and only if $m(G) \geq m(G')-k = \VPTsum(G')-n(G)-k \geq \VPTsum(G)-n(G)-k$.
\qed
\end{proof}

It is interesting that this formal relation, in addition to tree-depth, connects the class of trivially perfect graphs to another one of the four measures. 
Note that $\VPTmax$ (i.e. tree-depth) is also related to  trivially perfect completion where the objective is to minimize the clique number of the completed graph. This parallels definitions of the related graph parameters \emph{treewidth} and \emph{pathwidth} as the minimum clique number of any chordal or interval supergraph, respectively.
Nastos and Gao~\cite{Nastos} have indeed proposed to determine a specific notion of community structure in social networks, referred to as familial groups, via trivially perfect editing, i.e., by applying the minimum number of edge additions and edge removals to turn the graph into a trivially perfect graph. The generating tree of a closest trivially perfect graph is then interpreted as a vertex partition tree, and thus a hierarchical decomposition into nested communities that intersect at their cores. For both familial groups and $\VPTsum$, an imperfect structure is transformed into an idealized one, with the difference that $\VPTsum$ only allows for the addition of edges. Nastos and Gao~\cite{Nastos} prefer the restriction to addition when one is ``interested in seeing how individuals in a community are organized.''


Viewed from the opposite perspective, another interpretation of $\VPTsum$ is as a measure of network vulnerability under vertex removal. The capability of a network to withstand series of failures or attacks on its nodes is often assessed by observing changes in the size of the largest connected component, the reachability relation, or average distances~\cite{dong2019robustness}. An optimal vertex partition tree under $\VPTsum$ represents a worst-case attack scenario in which, for all vertices simultaneously, the average number of removals in their remaining component that it takes to detach a vertex is minimized.

The problem of adding the fewest edges to make a trivially perfect graph was shown NP-hard by Yannakakis in~\cite{YANNAKAKISM1981} and so Theorem~\ref{thm:VPTeqTPC} implies the following.

\begin{corollaryperso}\label{cor:nphard}
Computing $\VPTsum$ is NP-hard.
\end{corollaryperso}

\section{\texorpdfstring{$\VPTsum$ and $\EPTsum$ of trees}{VPT-sum and EPT-sum of Trees}}\label{sec:trees}


In this section we consider the case when the input graph 
$G$ is a tree. In this case, every minimal cut consists of one edge, and hence by Corollary \ref{cor:cut} the optimal clustering trees are edge partition trees, i.e. the internal nodes of $T$ are $E(G)$. This allows us to prove that the cut at any internal node $t$ of an optimal clustering tree is an internal edge of $G[t]$, unless $G[t]$ is a star, which in turn allows us to give an algorithm for caterpillars.

\begin{lemmaperso}\label{lem:leaveslast}
Let $T$ an optimal clustering tree of a tree $G$. For any internal node $t \in T$ 
with children $u,v$, if $G[t]$ is not a star, then neither $u$ nor $v$ are leaves in $T$. This implies that the edge associated with $t$ is an internal edge of $G$.
\end{lemmaperso}
\begin{proof}
In order to arrive at a contradiction, we assume that there exists a tree $G$ and an optimal hierarchical clustering $(T,\delta)$ where there is an internal node $t$ with children $u,v$, $G[t]$ is not a star, but $v$ is a leaf in $T$ (corresponding: $v_G$ is a leaf in $G[t]$). Since $G[t]$ is not a star, there must exist an internal edge
$e$ in $G[t]$. We can assume that none of the ancestors of $e_T$ in $T[t]$ also cut an internal edge in $G[t]$. Let $e'_T$ be the parent of $e_T$, and $a$ the sibling of $e_T$. Furthermore, let $b$ and $c$ be the children of $e_T$. Now, there are three cases:

\emph{Case 1}: $e$ is an internal edge in $G[e_T]$.
This is the simplest case. We know that both $b$ and $c$ are internal nodes in $T$, while $a$ must be a leaf in $T$. WLOG, we assume that the vertex $a_G$ has a neighbor in $G[b_T]$. Now,
$$n(G[a]) \cdot m(G[b,c]) = n(G[a]) = 1$$
and
$$n(G[c]) \cdot m(G[a,b]) = n(G[c]) > 1$$
Thus, by Lemma \ref{lemma:rotate}, $T$ cannot be optimal.

\emph{Case 2}: $e$ is not an internal edge in $G[e_T]$.
In this case, we assume WLOG that $b$ is a leaf in $T$. There are two subcases here:

\emph{Case 2a}: The vertex $a_G$ is neighbor with $b_G$.
This case is really identical with Case 1. By Lemma \ref{lemma:rotate}, we can perform a rotation to get a clustering with lower DC-cost than $(T,\delta)$, unless $c$ also is a leaf in $T$.
In this case, we switch $b$ and $c$ and end up in Case 2b.

\emph{Case 2b}: The vertex $a_G$ has a neighbor in $G[c]$.
In this case, we can perform an operation on $T$ to obtain a new clustering with no higher cost than $T$: Let $T'$ be such that the node labelled $b$ in $T$ is labelled $a$ in $T'$ and vice versa, and is otherwise identical for every node in $T$. Since the underlying tree structure has not been changed and $T$ is a clustering tree where all subtrees induce connected subgraphs, it follows that $\DCcost(G,T) = \DCcost(G,T)$. But we now observe that $e_G$ is cut at $e'_T$, the parent of $e_T$ in $T$. When considering $e'$ and $e''$, the parent of $e'$, we still have a Case 2 situation. If it is Case 2b, then we can perform the operation described above again. If it is Case 2a, then we can perform a rotation and get a clustering of \emph{strictly lower} $\DCcost$. Since $e$ is an internal edge in $G[t]$, there must exist some ancestor of $e'$ in $T$ that has a child whose associated vertex is a neighbor of $b_G$. We will therefore always end up in Case 2a at some point. As we have seen, in each of these cases we can find a hierarchical clustering with lower $\DCcost$ than $T$. But one of these cases must always be true! Therefore $T$ cannot be optimal and the observation is true as stated.
\qed\end{proof}
		
\begin{theoremperso}\label{thm:caterpillar}
The $\DCcost$ of a caterpillar can be computed in $O(n^3)$ time.
\end{theoremperso}
\begin{proof}
We view a caterpillar $G$ as a collection of stars $(X_1,\ldots,X_p)$ that are strung together. 
The central vertices of the stars $(x_1,\ldots,x_p)$ form the spine of $G$. Thus, every internal edge $x_ix_{i+1}$ in $G$ lies on the spine, and removing such an edge we get two sub-caterpillars, $(X_1,\ldots,X_i)$ and $(X_{i+1},\ldots,X_p)$. 
For every $i,j\in [p]$ with $i\leq j$, we define $DC[i,j]$ to be the $\DCcost$ of the sub-caterpillar $(X_i,X_{i+1},\dots,X_j)$.
Note that for a star $X$ on $n$ vertices we have $\DCcost(X) = \binom{n+1}{2}-1$ (one less than the $n$'th triangle number) as $\DCcost(K_1) = 0 = \binom{2}{2}-1$, and whichever edge we cut in a star on $n$ vertices we end up with a single vertex and a star on $n-1$ vertices.
Therefore $DC[i,i]=\binom{n(X_i)+1}{2}-1$ for every $i\in [p]$.
From Theorem \ref{thm:conn} and Lemma \ref{lem:leaveslast}, we deduce the following for every $i,j\in[p]$ with $i<j$.
\[DC[i,j]= \sum_{k\in [i,j]}n(X_k) + \min_{k\in \{i,i+1,...,j-1\}} DC[i,k] + DC[k+1,j]\]
Hence, to find $DC(G)$, we compute $DC[i,j]$ for every $i,j\in [p]$ with $i < j$ in order of increasing $j-i$ and return $DC[1,p]$.
For the runtime, note that calculating a cell $DC[i,j]$ takes time $O(n)$ and there are $O(n^2)$ such cells in the table.
\qed\end{proof}

This dynamic programming along the spine of a caterpillar can be generalized to compute the $\DCcost$ of any tree $G$ in time $n^{O(d_G)}$, where $d_G$ is the number of leaves of the spine-tree $G'$ of $G$. Note that for a caterpillar $G$ we have $d_G=2$.

In addition, we can show analogues of Theorem \ref{thm:conn} and Lemma \ref{lem:leaveslast} for $\VPTsum$, which enables us to form a polynomial-time algorithm for the $\VPTsum$ of caterpillars, and by the same generalization as above, for any tree $G$ where $d_G$ is bounded.

\begin{theoremperso}\label{thm:boundedspinetree}
$\DCcost$ and $\VPTsum$ of a tree $G$ is found in $n^{O(d_G)}$ time.
\end{theoremperso}
\begin{proof}
One of the simplest algorithms for calculating the $\DCcost$ of a tree $G$ is a DP algorithm that for each edge $uv$ recursively calculates the $\DCcost$ of $G_u$ and $G_v$ and finds the minimum over all edges. This algorithm clearly works because of Theorem \ref{thm:conn}. If $G$ has $N$ connected subtrees, then this algorithm takes time in $O(N\cdot n)$.

Our algorithm for the $\DCcost$ of caterpillars is essentially this simple algorithm, with one extra fact exploited: Lemma \ref{lem:leaveslast}, which implies that you can restrict your search to the connected subtrees of the spine-tree of $G$. Since there are only $O(n^2)$ connected subtrees of the spine of a caterpillar, the algorithm runs in $O(n^3)$ time. Let $G$ be any tree, and let $G'$ be the spine-tree $G$. Given $d_G=|L(G')|$, we can bound the number of connected subtrees of $G'$ as follows:

$G'$ can have at most $d_G-2$ vertices of degree 3 or more, and thus $G'$ has at most $2d_G-3$ paths between vertices of degree 3 or more. A connected subtree of $G'$ cannot contain 2 disjoint parts of such a path, thus when $d_G >2$ the number of connected subtrees of $G'$ is no more than $n^{2d_G-3}$ (for $d_G=2$ cutting the single path at any two positions gives a higher bound of $n^{2d_G-2}$).

Our algorithm, recursively computing the DC-cost of subtrees, thus takes time in $O(n^{2d_G-2}\cdot n)=n^{O(d_G)}$.

For the second result, on the $\VPTsum$ of trees, we observe that the same argument as above also must hold for $\VPTsum$, given the two crucial facts: Any subtree of an optimal VPT induces a connected subgraph of the input graph; and there are always an optimal VPT $T^*$ such that the leaves of the input tree are also leaves in $T^*$. Then, an algorithm that for each vertex $v\in G'$ recursively computes the $\VPTsum$ of every component in $G'-v$ and returns the minimum, effectively computes $\VPTsum(G)$ in $n^{O(d_G)}$ time. 

The first fact, analogous to Theorem \ref{thm:conn} for $\DCcost$, is true by definition: In \emph{any} VPT, the subtrees induce connected subgraphs. The second fact, analogous to Lemma \ref{lem:leaveslast}, is easy to prove:

Let $G$ be any tree, and let $T$ be any VPT of $G$. We assume that there are at least one leaf $v$ in $G$, such that the corresponding node $v_T$ is not a leaf in $T$. Since $G[v_T]$ is connected, the neighbor $u$ of $v$ is in $T[v]$. Furthermore, $v_T$ has only one child $c$ in $T$. We will now show that there exists another tree $T'$ where $v_T$ is a leaf, and $\VPTsum(G,T')\leq \VPTsum(G,T)$. We construct $T'$ from $T$ by moving down $v_T$, such that $c$ becomes a child of the parent of $v_T$, and $v_{T'}$ becomes a child of $u$. Then, $vd_{T'}(v_{T'}) = vd_T(u)$. So the depth of $v$ increases by $vd_T(u)-vd_T(v)$. But on the other hand, there are at least $vd_T(u)-vd_T(v)$ nodes in the subtree whose depth decreases by 1. Thus, $T'$ has no higher $\VPTsum$ than $T$.

From this fact, and the bound on connected subtrees of the spine-tree shown above, we conclude that the $\VPTsum$ of a caterpillar can be computed in $O(n^3)$ time, and more generally, the $\VPTsum$ of a tree can be computed in $n^{O(d_G)}$ time.
\qed\end{proof}

Lastly, we discuss the most well-studied approximation algorithm for $\DCcost$ \cite{Das16}, recursively partitioning the graph along a sparsest cut.
A \emph{sparsest cut} of a graph $G$ is a partition $(A,B)$ of $V(G)$ that minimizes the measure $\frac{m(G[A,B])}{|A|\cdot|B|}$. A sparsest cut must be a minimal cut.

For general graphs, finding a sparsest cut is NP-hard and must be approximated itself. On trees however, every minimal cut consists of one edge, and a sparsest cut is a \emph{balanced cut}, minimizing the size of the largest component. The optimal cut can therefore in trees be found efficiently. We call an edge of a tree \emph{balanced}, if it induces a balanced cut.

The results by Charikar and Chatziafratis~\cite{CC17} already indicate that the balanced cut algorithm gives an 8-approximation of the $\DCcost$ of trees (and other graph classes for which the sparsest cut can be found in polynomial time, like planar graphs, see~\cite{abboud2020new} for more information). In the following, we prove that for trees, we can guarantee a 2-approximation. 
We start by showing an upper bound on the $\DCcost$ of the two subtrees resulting from removing an arbitrary edge of a tree, and follow up with a stronger bound if the removed edge is balanced.

\begin{lemmaperso}\label{lemma:minless}
If $G$ is a tree and $e=uv\in E(G)$, then \\
\centerline{$\DCcost(G_u)+\DCcost(G_v) \leq \DCcost(G)-\min\{n(G_u),n(G_v)\}$.}
\end{lemmaperso}
\begin{proof}
Let $f=xy\in E(G)$ be the edge at the root of an optimal clustering tree $T$ of $G$. If $e=f$, then $\DCcost(G_u)+\DCcost(G_v)=\DCcost(G)-n$ by Theorem \ref{thm:eptsum} and we are done, so we can assume $e\neq f$ and w.l.o.g.\ that $e\in E(G_x)$. 

Suppose that $n(G_u)\leq n(G_v)$. Observe that $G_u$ is a subgraph of either $G_x$ or $G_y$, say $G_x$. A clustering tree $T_u$ of $G_u$ is obtained from $T$ as follows: if $T'_u$ is the minimal subtree of $T$ spanning the leaves mapped to vertices of $G_u$, then  contracting every path with degree-two internal nodes into a single edge results in $T_u$. A clustering tree $T_v$ of $G_v$ is obtained the same way.
We observe that, by construction, for every vertex $x\in V(G_v)$, $ed_{T_v}(x) \leq ed_{T}(x)$. Focusing on $G_u$ and $T_u$, as by assumption $G_u$ is a subgraph of $G_x$, $T_u$ can be obtained from the above process from the subtree $T[c_x]$ where $c_x$ is the child of the root node of $T$ such that $T[c_x]$ is a clustering tree of $G_x$. It follows that for every vertex $x\in V(G_u)$ we have $ed_{T_u}(x) +1\leq ed_{T}(x)$.
Using the previous inequalities and Definition~\ref{def:EPTsum}, we obtain:
\begin{align*}
  & \DCcost(G_u)+\DCcost(G_v)
   &\leq& ~\DCcost(G_u,T_u) + \DCcost(G_v,T_v)\\
  & = \sum_{x\in V(G_u)}\!\! ed_{T_u}(x) + \!\!\sum_{x\in V(G_v)}\!\! ed_{T_v}(x)
   &\leq& \sum_{x\in V(G_u)}\!\! (ed_T(x)-1) + \!\!\sum_{x\in V(G_v)}\!\! ed_T(x)\\
  & = \sum_{x\in V(G)}\! ed_T(x) - n(G_u)
   &=& ~\DCcost(G)-n(G_u)
  &\!\!\qed 
\end{align*}
\end{proof}

\begin{lemmaperso}\label{lemma:halfnless}
If $G$ is a tree and $e=uv\in E(G)$ balanced,
then $\DCcost(G_u)+\DCcost(G_v) \leq \DCcost(G)-\max\{n(G_u),n(G_v)\} \leq \DCcost(G)-\frac{n(G)}{2}$.
\end{lemmaperso}
\begin{proof}
The proof is by induction on the number of edges in the tree $G$.
The single edge $e=uv$ of a $K_2$ induces a balanced cut
and $\DCcost(G_u)+\DCcost(G_v)=0+0\leq 2-1=\DCcost(G)-\max\{n(G_u),n(G_v)\}$.

For the induction step assume that $G$ is a tree with at least two edges
and choose any balanced edge $e=uv$.
Let $f=xy$ be the edge at the root of an optimal clustering of $G$.
If $e=f$, then by definition $\DCcost(G)=\DCcost(G_u)+\DCcost(G_v)+n$ and the lemma holds. 
Assume therefore that $e\neq f$ and w.l.o.g.\ $f\in E(G_u)$ and $e\in E(G_y)$, as in Figure~\ref{fig:cut}.
We let $G_{uy}$ denote the subgraph induced by $V(G_u)\cap V(G_y)$.

\begin{figure}[tb]
    \centering
    \includegraphics[width=0.7\linewidth]{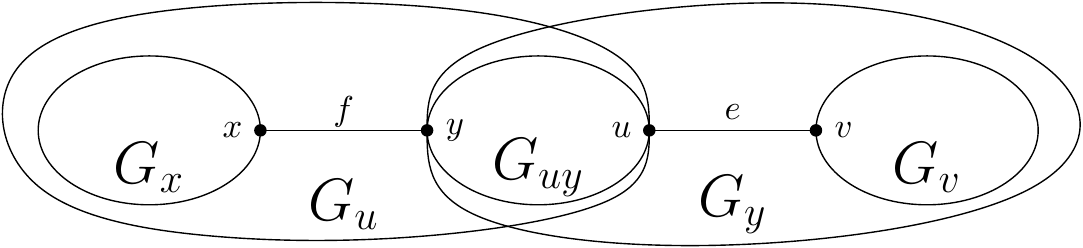}
    \caption{Inclusion relations between the subtrees $G_u,G_v,G_x,G_y$ and $G_{uy}$.}
    \label{fig:cut}
\end{figure}

By definition, $\DCcost(G_u)\leq\DCcost(G_x)+\DCcost(G_{uy})+(n-n(G_v))$,
and together with $\DCcost(G_{uy})+\DCcost(G_v)\leq\DCcost(G_y)-\eta$ where
$$
  \eta = \begin{cases}
    \max\{n(G_{uy}),n(G_v)\} & \text{if $e$ is balanced in $G_y$ (from induction hypothesis)}\\
    \min\{n(G_{uy}),n(G_v)\} & \text{otherwise (from Lemma~\ref{lemma:minless})}
  \end{cases}
$$
we get
$$\begin{array}{l}
  \DCcost(G_u)+\DCcost(G_v)\\
  \qquad\leq\quad\underset{=\DCcost(G)}{\underbrace{\DCcost(G_x)+\DCcost(G_y)+n}}
      -n(G_v)-\eta.~
\end{array}$$
It remains to show that
$n(G_v) + \eta \geq \max\{n(G_u),n(G_v)\}$, which is obvious if $n(G_u) \leq n(G_v)$. So we suppose that $n(G_u)>n(G_v)$ and proceed to show that $n(G_v)+\eta \geq n(G_u)$.
Since $e$ is balanced in $G$, we have by definition $$\max\{n(G_x)+n(G_{uy}), n(G_v)\}\leq \max\{n(G_x),n(G_{uy})+n(G_v)\},$$
implying that $ n(G_v)\geq n(G_x)$. It follows that 
$$n(G_v)+n(G_{uy})\geq n(G_x)+n(G_{uy})=n(G_u).$$
See Figure \ref{fig:cut}.
We have two cases to consider.
 If $n(G_{uy}) \leq n(G_v)$, then $n(G_{uy})=\min\{n(G_{uy}),n(G_v)\}\leq \eta$, implying that $n(G_v)+\eta\geq n(G_u)$.
 If $n(G_{uy}) > n(G_v)$, then $e$ is balanced in $G_y$, so $\eta$ is the maximum of $n(G_{v})$ and $n(G_{uy})$ and we are done since $n(G_u)\leq n(G_v)+\eta$.
%
\qed\end{proof}

Let a \emph{balanced clustering} of a tree $G$ be a clustering tree $\hat{T}$ such that for every internal node $e$ in $\hat{T}$, the edge corresponding to $e$ is a balanced edge in $G[e]$. As discussed earlier, a balanced clustering of a tree can be found efficiently. We now prove the guarantee of 2-approximation:

\begin{theoremperso}\label{thm:upperbound}
Let $G$ be a tree, and $\hat{T}$ a balanced clustering of $G$. We then have $\DCcost(G,\hat{T}) \leq 2\cdot\DCcost(G)$.
\end{theoremperso}
\begin{proof}
The overall proof goes by strong induction. For the base case, we easily see that for every tree on at most 2 vertices, all the balanced clustering trees are actually optimal; for these trees the statement follows trivially. For the induction step, we assume that for all trees on at most some $k$ vertices, the statement holds. We then look at a tree $G$ on $n = k+1$ vertices. We focus on two different clustering trees of $G$: $T^*$, which is an optimal clustering tree and has $\DCcost$ $W^*=\DCcost(G)$, and $\hat{T}$, which is a balanced clustering tree and has $\DCcost$ $\hat{W}$. Our aim is now to prove that $\hat{W} \leq 2\cdot W^*$.

We denote the root of $\hat{T}$ by $r=uv$ and its two children by $c_u$ and $c_v$. By definition, $\hat{T}[c_u]$ and $\hat{T}[c_v]$ are balanced clustering trees of $G_u$ and $G_v$, respectively. By our induction hypothesis, we know that $\DCcost(G_u,\hat{T}[c_u]) \leq 2\cdot\DCcost(G_u)$  and respectively for $c_v$. By definition we have  $\hat{W} = \DCcost(G_u,\hat{T}[c_u])+\DCcost(G_v,\hat{T}[c_v])+n$ which gives 
$\hat{W} \leq 2\cdot(\DCcost(G_u)+\DCcost(G_v))+n$. By Lemma \ref{lemma:halfnless} we have $\DCcost(G_u)+\DCcost(G_v) \leq \DCcost(G)-\frac{n}{2}$ and so $\hat{W} \leq 2\cdot(\DCcost(G)-\frac{n}{2})+n = 2\cdot W^*$ and we are done.
 \qed\end{proof}

On the other hand, the recursive sparsest cut algorithm will not necessarily compute the optimal value on trees. Actually, it fails already for caterpillars. 

\begin{theoremperso}\label{thm:ApproxLowBound}
 There is an infinite family of caterpillars $\{B_k \mid k\geq 3\}$ such that $\DCcost(B_k) \geq 2^k$ and for any balanced clustering tree $\hat{T}_k$ of $B_k$ has the property that $\DCcost(B_k,\hat{T}_k) / \DCcost(B_k) = 1 + \Omega(1/\sqrt{\DCcost(B_k)})$.
\end{theoremperso}
\begin{proof}
For any $k \geq 3$, we construct $B_k \in \mathcal{B}$ as a ``broomstick'', consisting of a path on $2^k$ vertices connected to a star on $\frac{3}{4}2^k$ vertices. If we first separate the path from the star and then take an optimal clustering for the path and the star, we get a clustering tree with cost
$$W^* = k\cdot 2^k + \frac{9}{32}2^{2k} + \frac{17}{8}2^k - 1$$
In other words, $\DCcost(B_k) \leq W^* \leq 2^{2k}$. Since the $\DCcost$ does not increase when taking subgraphs, we also have that $\DCcost(B_k) \geq \frac{9}{32}2^{2k}$.

Now, we discuss how a balanced clustering tree of $B_k$ can be made. There is one balanced edge in $B_k$: on the middle. We thus get a path on $\frac{7}{8}2^k$ vertices on one side. On the other side, where the rest of the path meets the star, the next balanced cut separates the path and the star. As all choices of balanced edges to further construct the clustering tree will lead to clustering trees with the same $\DCcost$, any balanced clustering tree of $B_k$ will thus have $\DCcost$
$$W' = k\cdot 2^k + \frac{9}{32}2^{2k} + \frac{20}{8}2^k - 1$$
We then see that $W'-W^* = \frac{3}{8}2^k = \frac{3}{14}n$. By the above inequalities, we get that
$$\frac{W'}{\DCcost(B_k)} \geq \frac{W'}{W^*} = 1 + \frac{W'-W^*}{W^*} \geq 1+\frac{\frac{3}{8}2^k}{2^{2k}} = 1+\frac{3}{8\cdot 2^k}$$
and
$$\sqrt{\DCcost(B_k)} \geq \sqrt{\frac{9}{32}2^{2k}} \geq \frac{1}{2}2^k$$
From these inequalities, we deduce that
$$\frac{W'}{\DCcost(B_k)} \geq 1 + \frac{1}{2\cdot \sqrt{\DCcost(B_k)}}$$
\qed\end{proof}

We conjecture that the actual performance of the balanced cut algorithm on trees is $1+O(1/\log n))$, which would substantially improve the 2-approximation ratio given by Theorem \ref{thm:upperbound}.

\section{Conclusion}

We have shown that two relatively new graph measures \textemdash\ the newly proposed, but intensively researched $\DCcost$, and the completely unexplored $\VPTsum$ \textemdash\ have close connections to two well-charted graph measures, edge and vertex ranking number. Furthermore, $\VPTsum$ is intimately related to the problem of \textsc{Trivially Perfect Completion}. We have also found simple descriptions of these measures in the form of recursive formulas. These observations are paired with some new results on $\DCcost$.

There is one subject which has not been touched upon: Just how tightly are these parameters related? As an example, take $\EPTmax$ and $\EPTsum$. It is obvious that for any graph $G$. $\EPTsum(G)$ must be upper-bounded by $\EPTmax(G)\cdot n(G)$. This bound is tight, as a path on $n = 2^k$ vertices for some $k$ has an optimal EPT that is a complete binary tree, with $\EPTmax$ equal to $k$ and $\EPTsum$ equal to $nk$. Similarly, $\VPTsum(G)$ is upper-bounded by $\VPTmax(G)\cdot n(G)$. We also believe that $\EPTsum(G)$ and $\VPTsum(G)$ should be lower-bounded by some fraction of $\EPTmax(G)\cdot n(G)$ and $\VPTmax(G)\cdot n(G)$ respectively, but what this fraction is, is not clear to us. Whether this is true or not, there may be some potential for exploiting techniques used for edge ranking to find out more about $\DCcost$ or vice versa. Much of the work done on $\DCcost$ has been about finding good approximation ratios for this problem (see e.g. \cite{CC17}, 
and \cite{Das16}). Comparing the performance of approximation algorithms for these two problems should be an interesting exercise.

Another question is about the recursive formulas. In our opinion, they form a simple and elegant description of the problems at hand. To our knowledge, formulas like these have not been studied in their own right. Are there other natural problems on graphs that can be captured by similar formulas?


\bibliography{ref}

\end{document}